\documentclass[letterpaper, 10 pt, conference]{ieeeconf}  

\IEEEoverridecommandlockouts                              
\usepackage{graphics} 
\usepackage{epsfig} 
\usepackage{mathptmx} 
\usepackage{times} 
\usepackage{amsmath} 
\usepackage{amssymb} 
\usepackage{xcolor}
\usepackage{multirow}
\usepackage{tabularx}

\usepackage{algorithm}
\usepackage{algorithmicx}
\usepackage{algpseudocode}

\DeclareMathAlphabet{\pazocal}{OMS}{zplm}{m}{n}
\newcommand{\norm}[1]{\left\lVert#1\right\rVert}

\newcommand{\rank}[1]{\mathrm{rank}\left(#1\right)}

\newcommand{\tr}[1]{\mathrm{tr}\left(#1\right)}

\newcommand{\diag}[1]{\mathrm{diag}\left(#1\right)}

\newtheorem{theorem}{Theorem}
\newtheorem{remark}{Remark}
\newtheorem{proposition}{Proposition}

\newtheorem{corollary}{Corollary}
\newtheorem{lemma}{Lemma}

\title{\LARGE \bf
Low-Rank Matrix Regression via Least-Angle Regression
}

\author{Mingzhou Yin and Matthias A. Müller%
\thanks{This work was supported by the Lower Saxony Ministry for Science and Culture within the program zukunft.niedersachsen.}%
\thanks{The authors are with the Institute of Automatic Control, Leibniz University Hannover, 30167 Hannover, Germany (e-mail: yin@irt.uni-hannover.de; mueller@irt.uni-hannover.de).}%
\thanks{© 2025 IEEE. Personal use of this material is permitted. Permission from IEEE must be obtained for all other uses, in any current or future media, including reprinting/republishing this material for advertising or promotional purposes, creating new collective works, for resale or redistribution to servers or lists, or reuse of any copyrighted component of this work in other works.}
}

\begin{document}

\maketitle
\thispagestyle{empty}
\pagestyle{empty}

\begin{abstract}
Low-rank matrix regression is a fundamental problem in data science with various applications in systems and control. Nuclear norm regularization has been widely applied to solve this problem due to its convexity. However, it suffers from high computational complexity and the inability to directly specify the rank. This work introduces a novel framework for low-rank matrix regression that addresses both unstructured and Hankel matrices. By decomposing the low-rank matrix into rank-1 bases, the problem is reformulated as an infinite-dimensional sparse learning problem. The least-angle regression (LAR) algorithm is then employed to solve this problem efficiently. For unstructured matrices, a closed-form LAR solution is derived with equivalence to a normalized nuclear norm regularization problem. For Hankel matrices, a real-valued polynomial basis reformulation enables effective LAR implementation. Two numerical examples in network modeling and system realization demonstrate that the proposed approach significantly outperforms the nuclear norm method in terms of estimation accuracy and computational efficiency.
\end{abstract}

\section{INTRODUCTION}

Following the principle of parsimony, finding low-order structures from data is critical in data science. Such low-order structures can often be interpreted as low-rank model matrices. Thus, data modeling tasks can be posed as estimating unknown matrices from noisy measurements, subject to rank conditions and possibly structural constraints. This problem is referred to as low-rank matrix regression. Low-rank matrix regression has a wide range of applications in systems and control, including system identification \cite{Smith_2014}, realization and model order reduction \cite{Schutter_2000}, and network modeling \cite{Basu_2019}, among others. It is also a central problem in machine learning \cite{Cohen_2015} and computer vision \cite{Yuan_2020}. See \cite{markovsky2012low,Markovsky_2008,Fazel_2013} for an overview.

One of the earliest and most well-known results in this regard is the Eckart-Young-Mirsky (EYM) theorem \cite{Eckart_1936}. It states that for unstructured matrices and when the regressor is an identity matrix, the closed-form solution to the low-rank matrix approximation problem is given by only keeping the most significant singular values corresponding to the rank of the unknown matrix. This method is known as truncated singular value decomposition (SVD). We refer to the identity regressor case as the approximation problem. Beyond the unstructured approximation problem, however, no closed-form solution or convex formulation exists in general.

Different algorithms have been proposed to obtain approximate solutions for low-rank matrix regression, including nonlinear optimization algorithms \cite{Markovsky_2014,Zvonarev_2021,park1999low}, singular spectrum analysis and mixed alternating projections \cite{golyandina2018singular,Cadzow_1988,Zvonarev_2022,Yin_2021,Wang_2019}, and convex relaxation \cite{Smith_2014,Fazel_2001,Fazel_2003}. Among them, this letter particularly focuses on the nuclear norm regularization approach, which provides the tightest convex surrogate for the rank function. This approach has gained popularity due to its versatility and ease of implementation, since the nuclear norm regularization problem can be reformulated as a semidefinite program (SDP) and solved using standard SDP solvers \cite{Fazel_2001}. However, despite being a convex problem, the nuclear norm approach has the following two drawbacks. 1) The computational complexity of SDPs scales unfavorably with the problem size, making it unsuitable for large-scale problems. 2) The rank of the estimates cannot be specified directly but relies on tuning the hyperparameter.

This work aims to improve the nuclear norm approach by presenting a closely related, yet more efficient and effective approach to low-rank matrix regression via least-angle regression (LAR). This approach works for both unstructured and Hankel matrices, and when the matrix is measured under a linear transformation with a general regressor. The Hankel matrix structure is of importance in various fields, including subspace identification \cite{Smith_2014}, behavioral system modeling \cite{markovsky2006exact}, signal processing \cite{Scharf_1991}, and image processing \cite{Kyong_Hwan_Jin_2015}. The proposed approach first reformulates the problem into an infinite-dimensional sparse learning problem by decomposing the low-rank matrix into a linear combination of rank-1 bases. The bases are selected as orthonormal ones for the unstructured case and polynomial ones for the Hankel case.

Then, the LAR algorithm \cite{Efron_2004} is applied to solve the sparse learning problems. LAR is a well-known algorithm in statistics for variable selection by constructing a solution path in the ``least-angle” direction of all active variables. The algorithm has a close connection with $l_1$-norm regularization or lasso, but can be implemented more efficiently with a similar computational complexity to least squares. However, its application to systems and control \cite{Zhang_2015,Chiuso_2012} is relatively limited. The main contribution of this work is as follows. 1) For the unstructured case, a closed-form LAR solution is derived. The solution is shown to be equivalent to a normalized nuclear norm regularization problem, generalizing the relation between LAR and lasso. 2) For the Hankel case, a modified LAR algorithm is proposed by using a real-valued reformulation of the polynomial basis. Two examples in network modeling and system realization are tested numerically. Numerical results demonstrate that the proposed LAR approaches perform significantly better than the nuclear norm solutions in terms of both estimation accuracy and computation time for both unstructured and Hankel matrices.

\textit{Notation.}
The complex conjugate of $z$ is denoted by $z^*$. The imaginary unit is denoted by $j$. The number of elements in a set $A$ is indicated by $\#(A)$. The notation $\mathbf{e}_i^n$ represents the unit vector along the $i$-th coordinate in $\mathbb{R}^n$. For a sequence $\left(x_i\right)_{i=1}^\infty$, the infinite Hankel operator is defined as $\mathcal{H}(x)$ with the $(i,k)$-th element being $x_{i+k-1}$, and the finite Hankel operator of depth $L$ is defined as
\begin{equation*}
    \mathcal{H}_L\left(x_{[m,n]}\right)=\begin{bmatrix}
        x_m&x_{m+1}&\cdots&x_{n-L+1}\\
        x_{m+1}&x_{m+2}&\cdots&x_{n-L+2}\\
        \vdots&\vdots&\ddots&\vdots\\
        x_{m+L-1}&x_{m+L}&\cdots&x_n\\
    \end{bmatrix}.
\end{equation*}

\section{PROBLEM STATEMENT AND BACKGROUND}
Consider the problem of estimating an unknown matrix $X\in\mathbb{M}^{m\times n}\subseteq\mathbb{R}^{m\times n}$ from a noisy linear measurement $Y\in\mathbb{R}^{p\times n}$: $Y = \Phi X+E$, where $\Phi\in\mathbb{R}^{p\times m}$ is the regressor that defines the measurement space, $E\in\mathbb{R}^{p\times n}$ denotes the noise matrix, and $\mathbb{M}^{m\times n}$ represents the matrix structure. We assume that $\Phi$ has full column rank with $\rank{\Phi}=m$, which requires $p\geq m$. The matrix $X$ can have an arbitrary shape with $\bar{n}:=\min(m,n)$. This work considers two types of matrix structures: the unstructured case where $\mathbb{M}^{m\times n}=\mathbb{R}^{m\times n}$ and Hankel matrices. The unknown matrix $X$ is known to have a low rank with $\rank{X}=r$, $r<\bar{n}$.

The estimation problem is usually solved by finding the best rank-$r$ matrix that fits the measurement $Y$:
\begin{equation}
    \hat{X}=\underset{X\in\mathbb{M}^{m\times n}}{\mathrm{arg\,min}}\ \norm{Y-\Phi X}_F^2\ \mathrm{s.t.}\ \rank{X}\leq r.
    \label{eq:slra}
\end{equation}
For the unstructured approximation problem, i.e., $\mathbb{M}^{m\times n}=\mathbb{R}^{m\times n}$ and $\Phi=\mathbb{I}_m$, let $Y=\sum_{i=1}^{\bar{n}}\tilde{\sigma}_i\tilde{\mathbf{u}}_i\tilde{\mathbf{v}}_i^\top$ be the SVD of $Y$, where $\tilde{\sigma}_i$ are the singular values in decreasing order and $\tilde{\mathbf{u}}_i\in\mathbb{R}^m$, $\tilde{\mathbf{v}}_i\in\mathbb{R}^n$ are the left and right singular vectors, respectively. Then, the EYM theorem \cite{Eckart_1936} shows that truncated SVD gives the closed-form solution to \eqref{eq:slra}, i.e., $\hat{X} = \sum_{i=1}^r \tilde{\sigma}_i\tilde{\mathbf{u}}_i\tilde{\mathbf{v}}_i^\top$. Unfortunately, apart from this special case, \eqref{eq:slra} is NP-hard in general.

A widely-used convex surrogate to \eqref{eq:slra} is the nuclear norm regularization, given by
\begin{equation}
    \hat{X}_\mathrm{nuc}=\underset{X\in\mathbb{M}^{m\times n}}{\mathrm{arg\,min}}\ \frac{1}{2}\norm{Y-\Phi X}_F^2+\lambda\norm{X}_*,
    \label{eq:nuc}
\end{equation}
where $\lambda$ is a hyperparameter that controls the rank of $\hat{X}_\mathrm{nuc}$ and $\norm{\cdot}_*$ denotes the nuclear norm, which is defined as the sum of all singular values. However, as discussed in the introduction, this approach is not suitable for large-scale problems, and obtaining an estimate of a specific rank $r$ requires tuning $\lambda$ by trial and error.

\subsection{Examples in Systems and Control}
\label{sec:exam}

This work considers two problems in systems and control as motivating examples.

Example 1 considers a low-rank network modeling problem \cite{Basu_2019}. Let $x_k\in\mathbb{R}^n$ be an $n$-dimensional time series in a network with a first-order vector autoregressive model $x_{k+1}=B^\top x_k+e_k$, where $B\in\mathbb{R}^{n\times n}$ is a low-rank transition matrix of rank $r$. Given $(p+1)$ consecutive observations $(x_k)_{k=1}^{p+1}$, the transition matrix $X=B$ can be estimated by solving the unstructured low-rank regression problem with $m=n$, $Y=\left[x_2\ \dots\ x_{p+1}\right]^\top$, and $\Phi=\left[x_1\ \dots\ x_p\right]^\top$.

Example 2 concerns a system realization problem. Consider a discrete-time, strictly causal, single-input, single-output, linear time-invariant system with an order-$r$ transfer function $G(q)=\sum_{i=1}^\infty g_i q^{-i}$, where $g_i$ is the impulse response. Suppose the impulse response has been estimated with $\hat{g}_k=g_k+e_k$ for $k=1,\dots,m+n-1$. Since the Hankel matrix of the impulse response has rank $r$, the impulse response can be estimated by solving the Hankel low-rank regression problem with $p=m$, $Y=\mathcal{H}_m\left(\hat{g}_{[1,m+n-1]}\right)$, $\Phi=\mathbb{I}_m$, $X=\mathcal{H}_m\left(g_{[1,m+n-1]}\right)$.

\subsection{Low-Rank Matrix Regression by Sparse Learning}
\label{sec:inf}

In this work, we study the problem by reformulating \eqref{eq:slra} into a sparse learning problem. In its standard form, sparse learning finds a sparse solution to a finite-dimensional regression problem by solving
\begin{equation}
    \hat{\sigma}=\underset{\sigma\in\mathbb{R}^{n_\sigma}}{\mathrm{arg\,min}}\ \norm{Y-{\textstyle\sum_{i=1}^{n_\sigma}} \sigma_i\tilde{X}_i}^2\ \mathrm{s.t.}\ \mathrm{card}(\sigma)\leq r,
    \label{eq:sl}
\end{equation}
where $r$ denotes the maximum number of active covariates, $\left(\tilde{X}_i\right)_{i=1}^{n_\sigma}$ denotes all the covariates with $n_\sigma\gg r$ being the number of covariates, and $\mathrm{card}(\cdot)$ denotes the cardinality function, which counts the number of nonzero elements.

We note that low-rank matrix regression can be regarded as finding a sparse combination of independent rank-1 matrices. This idea leads to the following infinite-dimensional sparse learning problem as a prototype:
\begin{equation}
    \begin{matrix}
    \left(\hat{X}_i,\hat{\sigma}_i\right)=&\underset{X_i,\sigma_i}{\mathrm{arg\,min}}&\norm{Y-\Phi\sum_{i=1}^r \sigma_i X_i}_F^2\\
    &\mathrm{s.t.}&X_i=\bar{\mathbf{u}}_i\bar{\mathbf{v}}_i^\top,\ i=1,\dots,r,\\
    &&\rank{\left[\bar{\mathbf{u}}_1\ \cdots\ \bar{\mathbf{u}}_r\right]}=r,\\
    &&\rank{\left[\bar{\mathbf{v}}_1\ \cdots\ \bar{\mathbf{v}}_r\right]}=r,\\
    &&\sum_{i=1}^r \sigma_i X_i\in\mathbb{M}^{m\times n},
    \end{matrix}
    \label{eq:infsl}
\end{equation}
where $\bar{\mathbf{u}}_i\in\mathbb{R}^{m}$, $\bar{\mathbf{v}}_i\in\mathbb{R}^{n}$, and $X_i$ constitutes an independent rank-1 basis. It is easy to see that \eqref{eq:infsl} is equivalent to \eqref{eq:slra}.
\begin{proposition}
    The optimal solution to \eqref{eq:slra} is given by $\hat{X}=\sum_{i=1}^r \hat{\sigma}_i\hat{X}_i$, where $\left(\hat{X}_i,\hat{\sigma}_i\right)$ is the optimal solution to \eqref{eq:infsl}.
    \label{prop:1}
\end{proposition}
\begin{proof}
    The rank condition $\rank{X}\leq r$ is satisfied iff there exist $\Sigma=\diag{\sigma_1,\dots,\sigma_r}$ and rank-$r$ matrices $\bar{U}=\left[\bar{\mathbf{u}}_1\ \cdots\ \bar{\mathbf{u}}_r\right]$, $\bar{V}=\left[\bar{\mathbf{v}}_1\ \cdots\ \bar{\mathbf{v}}_r\right]$, such that $X=\bar{U}\Sigma\bar{V}^\top=\sum_{i=1}^r \sigma_i X_i$. This corresponds to the first three constraints in \eqref{eq:infsl}.
\end{proof}
\begin{remark}
In general, the last constraint in \eqref{eq:infsl} is not equivalent to each selected rank-1 matrix satisfying the structure constraint. However, this property is important for applying LAR on structured matrices and will be shown for a particular type of Hankel matrix later in Section~\ref{sec:4}.
\end{remark}

Problem \eqref{eq:infsl} is very similar to the standard sparse learning problem \eqref{eq:sl} by considering the Frobenius norm and $\tilde{X}_i=\Phi X_i$, except that \eqref{eq:sl} considers a finite set of covariates $\left(\tilde{X}_i\right)_{i=1}^{n_\sigma}$, whereas \eqref{eq:infsl} considers an infinite set of rank-1 covariates. Such a problem is known as an infinite-dimensional sparse learning problem \cite{Yin_2022}.

\subsection{Least-Angle Regression}
\label{sec:LAR}

The sparse learning problem \eqref{eq:sl} is also NP-hard. In this regard, the least-angle regression approach provides a computationally efficient method to find approximate solutions to \eqref{eq:sl}. Unlike the nuclear norm approach, LAR obtains the complete solution path for all values of $r$ in one run. This subsection summarizes the standard procedure of LAR.

The idea of LAR is as follows, with a graphical illustration shown in Fig.~\ref{fig:1} for $p=2$, $n_\sigma=3$. For this demonstration, we assume that $Y,\tilde{X}_i\in\mathbb{R}^p$ are vector-valued, the $l_2$-norm is used in \eqref{eq:sl}, and $\tilde{X}_i$ has been normalized with $\norm{\tilde{X}_i}_2=1$. It starts with zero coefficients $\sigma=\mathbf{0}$ and finds the covariate $\tilde{X}_{i_1}$ ($\tilde{X}_{i_1}=\tilde{X}_3$ in Fig.~\ref{fig:1}) most correlated with the output $\mu_1=Y$, i.e., $i_1=\mathrm{arg\,max}_i\,|\tilde{X}_i^\top Y|$. We take the direction of this covariate $\zeta_1=\tilde{X}_{i_1}$ with a step size of $\eta_1$ until another covariate $\tilde{X}_{i_2}$ ($\tilde{X}_{i_2}=\tilde{X}_1$ in Fig.~\ref{fig:1}) is correlated with the model residual $\mu_2=Y-\eta_1\tilde{X}_{i_1}$ as much as $\tilde{X}_{i_1}$, i.e., $(i_2,\eta_1)=\mathrm{arg\,min}_{i,\eta}\,|\eta|\ \mathrm{s.t.}\ |\tilde{X}_i^\top (Y-\eta\tilde{X}_{i_1})|=|\tilde{X}_{i_1}^\top (Y-\eta\tilde{X}_{i_1})|$. Then the algorithm continues along the bisecting direction between $\tilde{X}_{i_1}$ and $\tilde{X}_{i_2}$, calculated by $\zeta_2=\mathrm{arg\,min}_{\zeta}\,\norm{\zeta}_2\ \mathrm{s.t.}\ \tilde{X}_{i_1}^\top\zeta=\tilde{X}_{i_2}^\top\zeta=1$, until a third covariate $\tilde{X}_{i_3}$ comes in with the same residual correlation as $\tilde{X}_{i_1}$ and $\tilde{X}_{i_2}$. Fig.~\ref{fig:1} ends at this step since the model residual is zero before the remaining covariate $\tilde{X}_2$ can join. If not, LAR proceeds along the ``least-angle direction" equiangular between $\tilde{X}_{i_1}$, $\tilde{X}_{i_2}$, and $\tilde{X}_{i_3}$ until the fourth covariate enters, and so on until all covariates are active or the model residual is zero.

A computationally efficient algorithm of this idea is described in \cite[Section~2]{Efron_2004}, which is summarized in Algorithm~\ref{al:1}, where $\tilde{X}=\left[\tilde{X}_1\ \cdots\ \tilde{X}_{n_\sigma}\right]$. The computational complexity of the algorithm for $k$ steps is $O(k^3+pk^2)$, which is 
the same as solving a least-squares problem with $k$ covariates \cite[Section~7]{Efron_2004}.

\begin{figure}[tb]
    \centering
    \includegraphics[width=0.8\columnwidth]{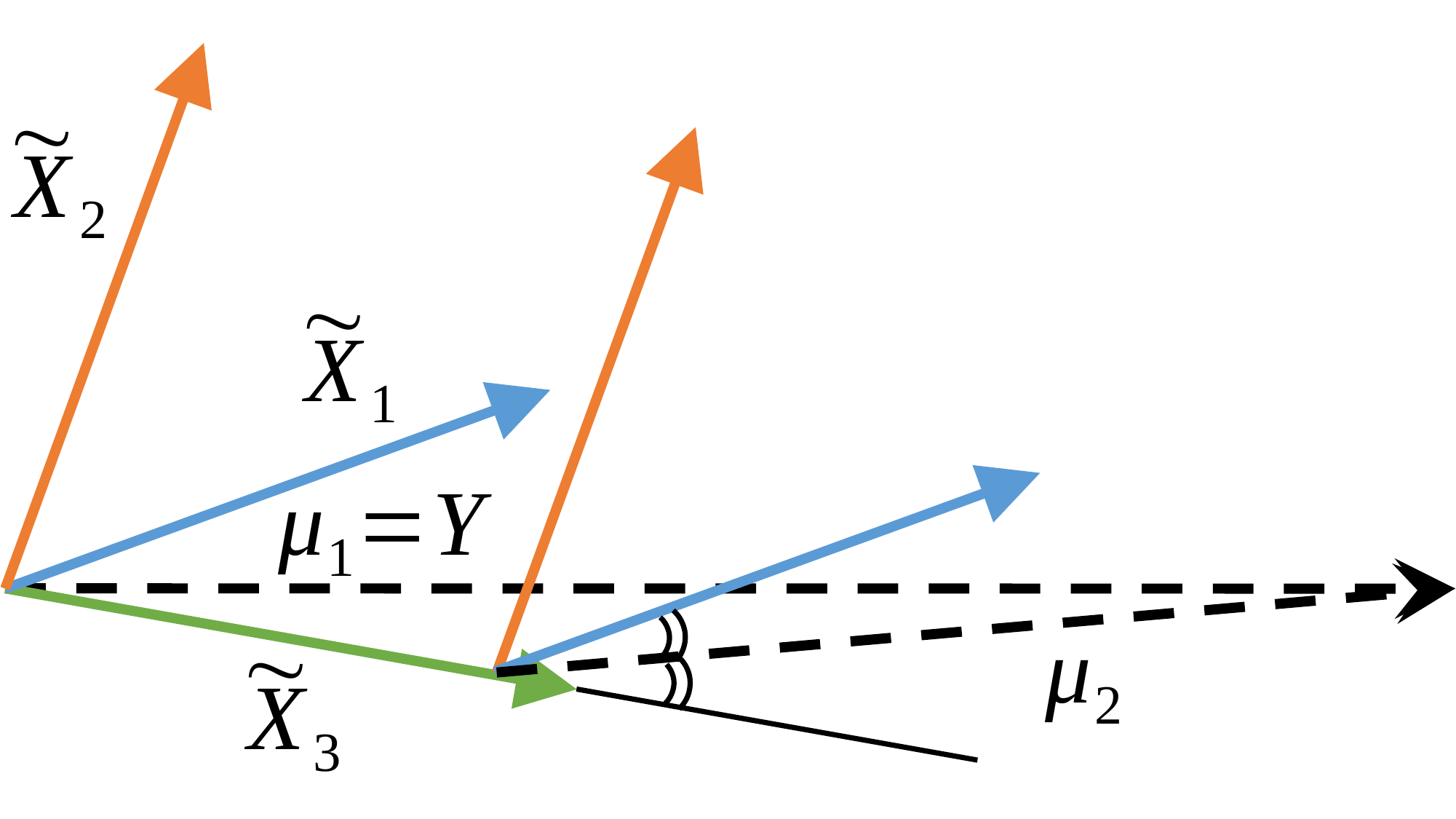}
    \vspace{-0.5em}
    
    \caption{Graphical illustration of least-angle regression. The solution first goes along $\tilde{X}_3$, and then the equiangular direction bisecting $\tilde{X}_3$ and $\tilde{X}_1$.}
    \label{fig:1}
\end{figure}

\begin{algorithm}[tb]
\caption{Standard least-angle regression}
    \begin{algorithmic}[1]
        \State \textbf{Initialization:} active set $\mathcal{A}_1=\mathrm{arg\,max}_k\,|c_k|$, $\mathbf{c}=\tilde{X}^\top Y$, prediction $\hat{Y}_1=\mathbf{0}$, parameter $\hat{\sigma}^1=\mathbf{0}$
        \For{$i=1,2,\dots$}
        \State Correlations on the residual: $\mathbf{c}=\tilde{X}^\top\mu_i$, where $\mu_i=Y-\hat{Y}_i$ is the model residual.
        \State Equiangular direction: $\zeta_i=\tilde{X}_{\mathcal{A}_i}\chi_i$, where
        $$\chi_i=\left(\tilde{X}_{\mathcal{A}_i}^\top\tilde{X}_{\mathcal{A}_i}\right)^{-1}\mathbf{1},\ \tilde{X}_{\mathcal{A}_i}=\left[\cdots\ \mathrm{sgn}(c_k)\tilde{X}_k\ \cdots\right]_{k\in \mathcal{A}_i}.$$
        \State Next covariate: $k_{i+1}=\underset{k\notin\mathcal{A}_i}{\mathrm{arg\,min}^+}\dfrac{\max(|\mathbf{c}|)\pm c_k}{1\pm a_k}$, where $\mathbf{a}=\tilde{X}^\top \zeta_i$ and arg\,min$^+$ indicates minimizing over positive components. Let $\eta_i$ be the minimum value obtained at $k=k_{i+1}$.
        \State $\mathcal{A}_{i+1}=\mathcal{A}_i\cup \left\{k_{i+1}\right\}$, $\hat{Y}_{i+1}=\hat{Y}_i+\eta_i\zeta_i$
        \State $\hat{\sigma}^{i+1}_{\mathcal{A}_i}=\left[\cdots\ \hat{\sigma}^{i+1}_{k}\ \cdots\right]^\top_{k\in \mathcal{A}_i}=\hat{\sigma}^i_{\mathcal{A}_i}+\eta_i\chi_i$
        \State $\hat{\sigma}^{i+1}_k=0$, for all $k\notin \mathcal{A}_i$
        \EndFor
    \end{algorithmic}
    \label{al:1}
\end{algorithm}

LAR is closely related to $l_1$-norm regularization or lasso, i.e., $\hat{\sigma}_\mathrm{lasso}(\lambda)=\mathrm{arg\,min}_{\sigma\in\mathbb{R}^{n_\sigma}}\,\frac{1}{2}\norm{Y-\sum_{i=1}^{n_\sigma} \sigma_i\tilde{X}_i}_2^2+\lambda\norm{\sigma}_1$.
With a minor modification to the algorithm, LAR provides all sparsity changing solutions of lasso, specified by $\mathrm{min}_\lambda\ \lambda$ $\mathrm{s.t.}\ \mathrm{card}\left(\hat{\sigma}_\mathrm{lasso}(\lambda)\right)=r$ for $r=1,2,\dots$ \cite[Section~3.1]{Efron_2004}. In particular, the modification states that whenever $\sigma_i$ flips its sign during the LAR algorithm, the algorithm pauses at $\sigma_i=0$, eliminates $\sigma_i$ from the active set, and continues with the equiangular direction without $\sigma_i$. However, it is not clear in the literature whether lasso solutions are better than pure LAR solutions in any sense. The modification also makes the sparsity non-monotonic along the solution path, which leads to ambiguity in the matrix regression of a particular rank. 

\section{UNSTRUCTURED MATRIX REGRESSION WITH LEAST-ANGLE REGRESSION}

This section generalizes the LAR algorithm presented in Section~\ref{sec:LAR} to the infinite-dimensional problem \eqref{eq:infsl} for the unstructured case where $\mathbb{M}^{m\times n}=\mathbb{R}^{m\times n}$. A closed-form solution is derived with equivalence to a normalized nuclear norm regularization problem.

For the unstructured case, it is natural to follow the idea of SVD and consider an orthonormal basis. Let $\Phi=U_\Phi S_\Phi V_\Phi^\top$ be the SVD of $\Phi$, where $U_\Phi\in\mathbb{R}^{p\times m}$ and $S_\Phi,V_\Phi\in\mathbb{R}^{m\times m}$. Since $\rank{\Phi}=m$, $S_\Phi$ is invertible. Then, we consider a special case of \eqref{eq:infsl} with $S_\Phi V_\Phi^\top X_i$ taking an orthonormal basis:
\begin{equation}
    \begin{matrix}
    \left(\hat{X}_i^\mathrm{u},\hat{\sigma}_i^\mathrm{u}\right)=&\underset{X_i^\mathrm{u},\sigma_i^\mathrm{u}}{\mathrm{arg\,min}}&\norm{Y-U_\Phi\sum_{i=1}^r \sigma_i^\mathrm{u} X_i^\mathrm{u}}_F^2\\
    &\mathrm{s.t.}&X_i^\mathrm{u}\in\mathcal{X}_i,\ i=1,\dots,r,
    \end{matrix}
    \label{eq:unssl}
\end{equation}
where
\begin{equation*}
\resizebox{\hsize}{!}{$
    \mathcal{X}_i=\left\{\mathbf{u}_i\mathbf{v}_i^\top\big|\,\norm{\mathbf{u}_i}_2=\norm{\mathbf{v}_i}_2=1,\mathbf{u}_i^\top\mathbf{u}_k=\mathbf{v}_i^\top\mathbf{v}_k=0,k=1,\dots,i-1\right\}
$}
\end{equation*}
constructs the set of rank-1 orthonormal matrices. Similar to Proposition~\ref{prop:1}, \eqref{eq:unssl} is equivalent to \eqref{eq:slra} for the unstructured case.
\begin{proposition}
    The optimal solution to \eqref{eq:slra} is given by $\hat{X}=V_\Phi S_\Phi^{-1}\sum_{i=1}^r \hat{\sigma}_i^\mathrm{u} \hat{X}_i^\mathrm{u}$, where $\left(\hat{X}_i^\mathrm{u},\hat{\sigma}_i^\mathrm{u}\right)$ is the optimal solution to \eqref{eq:unssl}, when $\mathbb{M}^{m\times n}=\mathbb{R}^{m\times n}$.
\end{proposition}
\begin{proof}
    The rank condition $\rank{X}\leq r$ is satisfied iff $S_\Phi V_\Phi^\top X$ has no more than $r$ non-zero singular values, i.e., $S_\Phi V_\Phi^\top X=\sum_{i=1}^r \sigma_i^\mathrm{u}\mathbf{u}_i\mathbf{v}_i^\top$, where $\mathbf{u}_i$ and $\mathbf{v}_i$ are orthonormal. This corresponds to the constraint in \eqref{eq:unssl}.
\end{proof}

Due to the favorable property of the orthonormality, the LAR solution to \eqref{eq:unssl} can be calculated in closed form, as shown in the following theorem.
\begin{theorem}
    Let $U_\Phi^\top Y=U^\mathrm{u}S^0(V^\mathrm{u})^{\!\top}$ be the SVD of $U_\Phi^\top Y$, where $U^\mathrm{u}=\left[\hat{\mathbf{u}}_1\ \cdots\ \hat{\mathbf{u}}_m\right]\in\mathbb{R}^{m\times m}$, $S^0\in\mathbb{R}^{m\times n}$, and $V^\mathrm{u}=\left[\hat{\mathbf{v}}_1\ \cdots\ \hat{\mathbf{v}}_n\right]\in\mathbb{R}^{n\times n}$. Let the $(i,i)$-th element of $S^0$ be $\sigma^0_i$. The LAR solution to \eqref{eq:unssl} is given by $\hat{X}_i^\mathrm{u}=\hat{\mathbf{u}}_i\hat{\mathbf{v}}_i^\top$ and $\hat{\sigma}_i^\mathrm{u}=\sigma^0_i-\sigma^0_{r+1}$.
    \label{thm:1}
\end{theorem}
\begin{proof}
We prove the theorem by induction. Similar to Section~\ref{sec:LAR}, we start with an empty model and find $X_1^\mathrm{u}\in\mathcal{X}_1$ that maximizes the correlation between $U_\Phi X_1^\mathrm{u}$ and the model residual $\mu_1=Y$ in terms of the Frobenius inner product:
\begin{equation*}
    \underset{X_1^\mathrm{u}\in\mathcal{X}_1}{\mathrm{arg\,max}}\ \frac{\left<U_\Phi X_1^\mathrm{u},Y\right>_F}{\norm{U_\Phi X_1^\mathrm{u}}_F\norm{Y}_F}=\underset{X_1^\mathrm{u}\in\mathcal{X}_1}{\mathrm{arg\,max}}\,\left<U_\Phi X_1^\mathrm{u},Y\right>_F,
\end{equation*}
since $\norm{U_\Phi X_1^\mathrm{u}}_F=1$. From the definition of $\mathcal{X}_1$, the problem is equivalent to
\begin{equation*}
\resizebox{\hsize}{!}{$
    \underset{\norm{\mathbf{u}_1}_2=\norm{\mathbf{v}_1}_2=1}{\mathrm{arg\,max}}\tr{\mathbf{v}_1\mathbf{u}_1^\top U_\Phi^\top Y}=\underset{\norm{\mathbf{u}_1}_2=\norm{\mathbf{v}_1}_2=1}{\mathrm{arg\,max}}\mathbf{u}_1^\top U_\Phi^\top Y\mathbf{v}_1
    =\left(\hat{\mathbf{u}}_1,\hat{\mathbf{v}}_1\right).
$}
\end{equation*}

In the first iteration, we go along the direction of $\zeta_1=U_\Phi X_1^\mathrm{u}$ and find the step size $\eta_1$ along $\zeta_1$ such that the new model residual $\mu_2=\mu_1-\eta_1 U_\Phi X_1^\mathrm{u}$ correlates with $U_\Phi X_1^\mathrm{u}$ as much as a new covariate $U_\Phi X_2^\mathrm{u}$. Note that $U_\Phi^\top\mu_2=U_\Phi^\top\mu_1-\eta_1X_1^\mathrm{u}=U^\mathrm{u}S^1(V^\mathrm{u})^{\!\top}$, where $S^1=S^0-\eta_1\mathbf{e}_1^m\left(\mathbf{e}_1^n\right)^\top$. We have $\left<U_\Phi X_1^\mathrm{u},\mu_2\right>_F=\sigma^0_1-\eta_1$. So $\eta_1$ is selected such that
\begin{align*}
\begin{matrix}
    \sigma^0_1-\eta_1=\underset{X_2^\mathrm{u}\in\mathcal{X}_2}{\mathrm{max}}\left<U_\Phi X_2^\mathrm{u},\mu_2\right>_F=\!\!\!\!\!\!&
    \underset{\mathbf{u}_2,\mathbf{v}_2}{\mathrm{max}}&\!\!\!\!\!\!\mathbf{u}_2^\top U^\mathrm{u}S^1(V^\mathrm{u})^{\!\top}\mathbf{v}_2\\
    &\mathrm{s.t.}&\!\!\!\!\!\!\norm{\mathbf{u}_2}_2=\norm{\mathbf{v}_2}_2=1,\\
    &&\!\!\!\!\!\!\mathbf{u}_2^\top\mathbf{u}_1=\mathbf{v}_2^\top\mathbf{v}_1=0.
\end{matrix}
\end{align*}%
Since $\mathbf{u}_2$, $\mathbf{v}_2$ are perpendicular to $\hat{\mathbf{u}}_1$, $\hat{\mathbf{v}}_1$, respectively, the maximum of the right-hand side is $\sigma^0_2$, obtained when $\mathbf{u}_2=\hat{\mathbf{u}}_2$ and $\mathbf{v}_2=\hat{\mathbf{v}}_2$. This leads to $\eta_1=\sigma^0_1-\sigma^0_2$.

Suppose $\mathbf{u}_i=\hat{\mathbf{u}}_i$, $\mathbf{v}_i=\hat{\mathbf{v}}_i$ for all $i=1,\dots,k$ and $\eta_i=\sigma^0_i-\sigma^0_{i+1}$ for all $i=1,\dots,k-1$. Since all $X_i^\mathrm{u}$'s are perpendicular, the equiangular direction at the $k$-th iteration is given by $\zeta_k=U_\Phi \sum_{i=1}^k X_i^\mathrm{u}$. Following a similar procedure as the first iteration, we have $\mu_{k+1}=\mu_k-\eta_k\zeta_k$, $S^k=S^{k-1}-\eta_k\sum_{i=1}^k\mathbf{e}_i^m\left(\mathbf{e}_i^n\right)^\top$, and
\begin{equation*}
\begin{matrix}
    \sigma^0_k-\eta_k&=&\underset{\mathbf{u}_{k+1},\mathbf{v}_{k+1}}{\mathrm{max}}&\mathbf{u}_{k+1}^\top U^\mathrm{u}S^k(V^\mathrm{u})^{\!\top}\mathbf{v}_{k+1}\\
    &&\mathrm{s.t.}&\norm{\mathbf{u}_{k+1}}_2=\norm{\mathbf{v}_{k+1}}_2=1,\\
    &&&\mathbf{u}_{k+1}^\top\mathbf{u}_i=\mathbf{v}_{k+1}^\top\mathbf{v}_i=0,\,i=1,\dots,k.
\end{matrix}
\end{equation*}
This leads to $\mathbf{u}_{k+1}=\hat{\mathbf{u}}_{k+1}$, $\mathbf{v}_{k+1}=\hat{\mathbf{v}}_{k+1}$ and $\eta_k=\sigma^0_k-\sigma^0_{k+1}$.

Finally, we note $\hat{\sigma}_i^\mathrm{u}=\sum_{k=i}^r \eta_k=\sigma^0_i-\sigma^0_{r+1}$, which completes the proof.
\end{proof}

Therefore, the LAR algorithm gives the following rank-$r$ solution for unstructured low-rank matrix regression:
\begin{equation}
    \hat{X}_\mathrm{LAR}=V_\Phi S_\Phi^{-1}\sum_{i=1}^r \left(\sigma^0_i-\sigma^0_{r+1}\right)\hat{\mathbf{u}}_i\hat{\mathbf{v}}_i^\top.
    \label{eq:xlar}
\end{equation}

It has been well-known that for the unstructured approximation case with $\Phi=\mathbb{I}_m$, \eqref{eq:nuc} admits a closed-form solution \cite[Theorem~2.1]{Cai_2010}: $\hat{X}_\mathrm{nuc}=\sum_{i=1}^{\bar{n}}\max{\left(\tilde{\sigma}_i-\lambda,0\right)}\tilde{\mathbf{u}}_i\tilde{\mathbf{v}}_i^\top$. Using the close relation between LAR and lasso, this result can be generalized to regression problems with a general $\Phi$.
\begin{corollary}
Consider the normalized nuclear norm regularization problem:
\begin{equation}
    \hat{X}_\mathrm{nuc,n}(\lambda)=\underset{X\in\mathbb{R}^{m\times n}}{\mathrm{arg\,min}}\ \frac{1}{2}\norm{Y-\Phi X}_F^2+\lambda\norm{S_\Phi V_\Phi^\top X}_*.
    \label{eq:nucnorm}
\end{equation}
The best rank-$r$ solution, i.e., $\mathrm{min}_\lambda\ \lambda\ \mathrm{s.t.}\ \rank{\hat{X}_\mathrm{nuc,n}(\lambda)}=r$ is given by $\hat{X}_\mathrm{LAR}$ in \eqref{eq:xlar}.
\end{corollary}
\begin{proof}
    Note that \eqref{eq:nucnorm} is equivalent to the lasso version of \eqref{eq:unssl}: $\mathrm{min}_{\left(X_i^\mathrm{u},\sigma_i^\mathrm{u}\right)}\ \frac{1}{2}\norm{Y-U_\Phi\sum_{i=1}^{\bar{n}} \sigma_i^\mathrm{u} X_i^\mathrm{u}}_F^2 + \lambda\sum_{i=1}^{\bar{n}} \sigma_i^\mathrm{u}\ \mathrm{s.t.}\ X_i^\mathrm{u}\in\mathcal{X}_i,\ i=1,\dots,\bar{n}.$
    As can be seen from the proof of Theorem~\ref{thm:1}, the coefficients of the active covariates are monotonically increasing, so the lasso modification in \cite[Section~3.1]{Efron_2004} cannot be triggered. Therefore, \cite[Theorem~1]{Efron_2004} completes the proof.
\end{proof}

\section{HANKEL MATRIX REGRESSION WITH LEAST-ANGLE REGRESSION}
\label{sec:4}

This section further considers the Hankel structure constraint in \eqref{eq:infsl} by employing a basis of rank-1 Hankel matrices $X_i$ and proposes a modified LAR algorithm for Hankel matrix regression. Unlike the unstructured case, orthonormal bases cannot be employed since the SVD of a Hankel matrix is not Hankel. Instead, complex polynomial bases are considered. With a slight abuse of notation, define $\bar{\mathbf{u}}_z=\left[1\ z\ \dots\ z^{m-1}\right]^\top,\ \bar{\mathbf{v}}_z=\left[1\ z\ \dots\ z^{n-1}\right]^\top,\ X_z=\bar{\mathbf{u}}_z\bar{\mathbf{v}}_z^\top$, where $z\in\mathbb{C}$. We would like to show that any rank-$r$ Hankel matrix $X$ can be expressed as $X=\sum_{i=1}^r \sigma_{z_i}X_{z_i}$, where $\sigma_{z_i}\in\mathbb{C}$. However, this statement is not true. A trivial counterexample is $X=\begin{bmatrix}
    0&0\\0&1
\end{bmatrix}$, which has rank 1 but cannot be expressed as $\bar{\mathbf{u}}_z\bar{\mathbf{v}}_z^\top$. Thus, this section focuses on a particular type of Hankel matrices that can be decomposed with polynomial bases, specified by the following lemmas.
\begin{lemma}
    An infinite Hankel matrix $\mathcal{H}(x)$ satisfies $\rank{\mathcal{H}(x)}\leq r$ iff $\left(x_k\right)_{k=1}^\infty$ satisfies the homogeneous linear difference equation $x_{k+r}+\sum_{i=0}^{r-1} a_i x_{k+i}=0$ for all $k\geq 1$.
    \label{lm:1}
\end{lemma}
\begin{proof}
    Let $s_i$ be the $i$-th column of $\mathcal{H}(x)$. Suppose $s_{k+1}$ is dependent on $s_1,\dots,s_k$. Then, due to the shift property of Hankel matrices, $s_i$ is dependent on $s_{i-k},\dots,s_{i-1}$ for all $i\geq k+1$. This is equivalent to $\rank{\mathcal{H}(x)}\leq k$. Therefore, $\rank{\mathcal{H}(x)}\leq r$ is equivalent to $s_{r+1}$ depending on $s_1,\dots,s_r$. The latter is further equivalent to $x_{k+r}+\sum_{i=0}^{r-1} a_i x_{k+i}=0$ for all $k\geq 1$.
\end{proof}
\begin{lemma}
    A Hankel matrix $X=\mathcal{H}_m\left(x_{[1,m+n-1]}\right)$ with $\rank{X}\leq r$ can be decomposed as $X=\sum_{i=1}^r \sigma_{z_i}X_{z_i}$ if 1) $X$ is a submatrix of an infinite Hankel matrix $\mathcal{H}(x)$ with $\rank{\mathcal{H}(x)}\leq r$, and 2) the characteristic polynomial $p(q)=q^r+\sum_{i=0}^{r-1} a_i q^i$ of the linear difference equation $x_{k+r}+\sum_{i=0}^{r-1} a_i x_{k+i}=0$ has no repeated roots.
    \label{lm:3}
\end{lemma}
\begin{proof}
    The decomposition $X=\sum_{i=1}^r \sigma_{z_i}X_{z_i}$ is equivalent to $x_k=\sum_{i=1}^r \sigma_{z_i} z_i^{k-1}$ for $k=1,\dots,m+n-1$. From Lemma~\ref{lm:1}, condition 1) implies that $x_{[1,m+n-1]}$ is a subsequence of $\left(x_k\right)_{k=1}^\infty:x_{k+r}+\sum_{i=0}^{r-1} a_i x_{k+i}=0$ for all $k\geq 1$. Let the non-repeated roots of $p(q)$ be $\left(z_i\right)_{i=1}^r$. The general solution of the linear difference equation \cite[Theorem~3.6]{kelley2001difference} is given by $x_k=\sum_{i=1}^r \sigma_{z_i} z_i^{k-1}$, which completes the proof.
\end{proof}

The first condition is naturally satisfied in most problems in systems and control, such as Example~2 in Section~\ref{sec:exam}. The second condition is often assumed in partial fraction expansion analysis, such as \cite{Gragg_1989,Shah_2012}, and repeated poles can be approximated by a combination of close distinct poles with arbitrarily high accuracy.

Since real-valued matrices $X$ are considered, if $X_z$ is in the decomposition, $X_{z^*}$ should also be included with $\sigma_{z^*}=\sigma_z^*$. Let $z=|z|\exp{(j\theta)}$ and $\sigma_z=|\sigma_z|\exp{(j\psi)}$. We have $\sigma_z z+\sigma_z^* z^*=2|\sigma_z||z|\cos{(\psi+\theta)}$. So, the decomposition can be reparametrized with real-valued modes $\sigma_{z_i} X_{z_i}+\sigma_{z_i}^* X_{z_i^*}=2|\sigma_{z_i}|X_{z_i}^{\psi_i}$, where
\begin{equation}
    X_{z_i}^{\psi_i}=\mathcal{H}_m\left(\xi_{[1,m+n-1]}^i\right),\ \xi_k^i=|z_i|^{k-1}\cos{(\psi_i+(k-1)\theta_i)}.
    \label{eq:xzpsi}
\end{equation}
Note that $X_{z_i}^{\psi_i}$ has rank 1 if $z_i$ is real and rank 2 if $z_i$ is complex. This observation provides a real-valued reformulation of \eqref{eq:slra} for $X$ that satisfies the conditions in Lemma~\ref{lm:3} by considering $z_i$ in the closed upper half plane with $\theta_i\in[0,\pi]$.
\begin{proposition}
    Let $\mathbb{M}^{m\times n}$ be the set of Hankel matrices that satisfy conditions 1) and 2) in Lemma~\ref{lm:3}. The optimal solution to \eqref{eq:slra} is given by $\hat{X}=\sum_{i=1}^{n_r} \hat{\sigma}_{z_i}'\hat{X}_{z_i}^{\psi_i}$, where $\left(\hat{X}_{z_i}^{\psi_i},\hat{\sigma}_{z_i}'\right)$ is given by
    \begin{equation}
    \begin{matrix}
    \underset{X_{z_i}^{\psi_i},\sigma_{z_i}'}{\mathrm{arg\,min}}&\norm{Y-\Phi\sum_{i=1}^{n_r} \sigma_{z_i}' X_{z_i}^{\psi_i}}_F^2\\
    \mathrm{s.t.}&\eqref{eq:xzpsi},\ \theta_i\in[0,\pi],\ i=1,\dots,n_r,\\
    &z_i\neq z_j,\ \forall\, i\neq j,\ n_\mathrm{real}+2n_\mathrm{comp}=r,
    \end{matrix}
    \label{eq:hsl}
\end{equation}
where $n_\mathrm{real},n_\mathrm{comp}$ are the number of real and complex $z_i$'s, respectively.
\end{proposition}
\begin{proof}
    According to Lemma~\ref{lm:3}, $X$ can be decomposed as $X=\sum_{i=1}^r \sigma_{z_i}X_{z_i}$. Without loss of generality, assume that $z_i$ is real for $i=1,\dots,n_\mathrm{real}$ and $z_i=z_{i+n_r-n_\mathrm{real}}^*$, $\theta_i\in(0,\pi)$ for $i=n_\mathrm{real}+1,\dots,n_r$. Then, we have
    \begin{align*}
        X&={\textstyle \sum_{i=1}^r}\,\sigma_{z_i}X_{z_i}={\textstyle \sum_{i=1}^{n_\mathrm{real}}}\,\sigma_{z_i}X_{z_i} + {\textstyle \sum_{i=n_\mathrm{real}+1}^{n_r}}\,\sigma_{z_i}X_{z_i}+\sigma_{z_i}^* X_{z_i^*}\\
        &={\textstyle \sum_{i=1}^{n_\mathrm{real}}}\,|\sigma_{z_i}|X_{z_i}^{\psi_i}+{\textstyle \sum_{i=n_\mathrm{real}+1}^{n_r}}\,2|\sigma_{z_i}|X_{z_i}^{\psi_i}.
    \end{align*}
    Selecting $\sigma_{z_i}'=|\sigma_{z_i}|$ for $i=1,\dots,n_\mathrm{real}$ and $\sigma_{z_i}'=2|\sigma_{z_i}|$ for $i=n_\mathrm{real}+1,\dots,n_r$ completes the proof.
\end{proof}
\begin{remark}
    If we focus on the case where $\left(x_k\right)_{k=1}^\infty$ is bounded, $z_i$ can be further constrained with $|z_i|\leq 1$. This is useful when analyzing stable systems.
    \label{rm:1}
\end{remark}

Then, we are ready to apply the LAR algorithm to \eqref{eq:hsl}. Unfortunately, a closed-form solution like Theorem~\ref{thm:1} does not exist for the Hankel case, but a similar procedure to the proof of Theorem~\ref{thm:1} can be adopted, which is summarized in Algorithm~\ref{al:2}. In detail, line 1 initializes the algorithm by finding the covariate most correlated with the output; line 3 finds the equiangular direction among the active covariates; line 4 calculates the step size before a new covariate is added; line 5 updates the variables with the new covariate.

Algorithm~\ref{al:2} only considers positive correlations without loss of generality, since one can take $X_z^{\psi+\pi}=-X_z^\psi$. Note that $X_z^\psi$ can be decomposed as $X_z^\psi=U_z R_\psi V_z^\top$, where
\begin{align*}
    U_z&=\begin{bmatrix}
        1&\phantom{-}|z|\cos{\theta}&\dots&\phantom{-}|z|^{m-1}\cos{\left((m-1)\theta\right)}\\
        0&-|z|\sin{\theta}&\dots&-|z|^{m-1}\sin{\left((m-1)\theta\right)}
    \end{bmatrix}^\top,\\
    R_\psi&=\begin{bmatrix}
        \cos{\psi}&-\sin{\psi}\\
        \sin{\psi}&\phantom{-}\cos{\psi}
    \end{bmatrix},\\
    V_z&=\begin{bmatrix}
        1&|z|\cos{\theta}&\dots&|z|^{n-1}\cos{\left((n-1)\theta\right)}\\
        0&|z|\sin{\theta}&\dots&|z|^{n-1}\sin{\left((n-1)\theta\right)}
    \end{bmatrix}^\top.
\end{align*}
For computational efficiency, the Frobenius inner product can be calculated as $\left<\Phi X_z^\psi, \Gamma\right>_F=\tr{U_z^\top\Phi^\top\Gamma V_zR_\psi^\top}$. Algorithm~\ref{al:2} requires solving non-convex optimization problems in lines 1 and 4. However, the dimension of the optimization variable is small and does not depend on the problem size. In line 4, $\epsilon>0$ is a small constant that guarantees the selection of a distinct $z_{i+1}$.
\begin{algorithm}[tb]
\caption{Least-angle regression for Hankel matrix regression}
    \begin{algorithmic}[1]
        \State \textbf{Initialization:} active set $\mathcal{A}_1=\left\{\left(z_1,\psi_1\right)\right\}$, where
        \begin{equation*}
            \left(z_1,\psi_1\right)=\underset{z,\psi}{\mathrm{arg\,max}}\ \left<\Phi X_z^\psi,Y\right>_F\big/\norm{\Phi X_z^\psi}_F\ \mathrm{s.t.}\ \theta\in[0,\pi],
        \end{equation*}
        prediction $\hat{Y}_1=\mathbf{0}$, parameter $\hat{\sigma}_{z_1}'^{,1}=0$
        \For{$i=1,2,\dots$}
        \State Equiangular direction: find $Z_i=\sum_{\left(z,\psi\right)\in\mathcal{A}_i}\chi_z^\psi \Phi X_z^\psi$ parameterized by $\left\{\chi_z^\psi\in\mathbb{R}\,|\left(z,\psi\right)\in\mathcal{A}_i\right\}$, such that $\left<\Phi X_z^\psi,Z_i\right>_F=\norm{\Phi X_z^\psi}_F$, for all $\left(z,\psi\right)\in\mathcal{A}_i$.
        \State Next covariate: $\left(z_{i+1},\psi_{i+1},\eta_i\right)$ are obtained by
        \begin{equation*}
            \begin{matrix}
            \underset{z,\psi,\eta}{\mathrm{arg\,min}}&\eta\\
            \mathrm{s.t.}&\frac{\left<\Phi X_z^\psi,\,Y-\hat{Y}_i-\eta Z_i\right>_F}{\norm{\Phi X_z^\psi}_F}=\frac{\left<\Phi X_{z_1}^{\psi_1},\,Y-\hat{Y}_i-\eta Z_i\right>_F}{\norm{\Phi X_{z_1}^{\psi_1}}_F},\\
            &\eta\geq 0,\ |z-z'|\geq \epsilon,\ \forall\left(z',\cdot\right)\in\mathcal{A}_i,\ \theta\in[0,\pi].
            \end{matrix}
        \end{equation*}
        \State $\mathcal{A}_{i+1}=\mathcal{A}_i\cup \left\{\left(z_{i+1},\psi_{i+1}\right)\right\}$, $\hat{Y}_{i+1}=\hat{Y}_i+\eta_iZ_i$, $\hat{\sigma}_{z}'^{,i+1}=\hat{\sigma}_{z}'^{,i}+\eta_i\chi_{z}^{\psi}$ for all $\left(z,\psi\right)\in\mathcal{A}_i$, $\hat{\sigma}_{z_{i+1}}'^{,i+1}=0$
        \EndFor
    \end{algorithmic}
    \label{al:2}
\end{algorithm}

The step size selection of $\eta_i$ in Algorithm~\ref{al:2} is conservative to allow the inclusion of further covariates, which induces a negative bias. When a desired rank is achieved, the estimate can be debiased by solving a least-squares problem:
\begin{equation}
    \begin{matrix}
    \underset{\sigma_{z}'\in\mathbb{R}^{n_r}}{\mathrm{min}}\ \norm{Y-\Phi\sum_{i=1}^{n_r} \sigma_{z_i}' X_{z_i}^{\psi_i}}_F^2
    \end{matrix}
    \label{eq:debias}
\end{equation}
where $z_i$ and $\psi_i$ are given by Algorithm~\ref{al:2}.

\section{NUMERICAL EXAMPLES}
This section demonstrates the performance of the LAR solutions \eqref{eq:xlar} and Algorithm~\ref{al:2} and compares them with the nuclear norm solution \eqref{eq:nuc}, for both unstructured and Hankel matrix regression.\footnote{The codes are available at https://doi.org/10.25835/dx2jik1l.} In both examples, 120 Monte Carlo simulations are conducted with the error $e_k$ being zero-mean and i.i.d. Gaussian. The convex programs are solved by YALMIP and Mosek, whereas the non-convex programs are solved by the \textsc{Matlab} function \texttt{fmincon} with the interior point method. The estimation accuracy is evaluated by $\norm{\hat{X}-X_0}_F^2$, where $X_0$ is the true low-rank matrix.

For unstructured matrix regression, Example~1 in Section~\ref{sec:exam} is considered with $m=n=40$, $p=80$, and $r=10$. The rank-$r$ transition matrix $B$ is generated by $B=b_0B_1B_2^\top$, where $B_1,B_2\in\mathbb{R}^{n\times r}$ contain unit i.i.d. Gaussian entries and $b_0$ is selected such that the spectral radius of $B$ is 0.95. The standard deviation of $e_k$ is 0.01. We also compare a two-step approach of first estimating $\hat{X}_\mathrm{LS}$ by least squares and subsequently applying truncated SVD of rank $r$. This approach is referred to as \textit{LS-TSVD}. To obtain a specific rank $r$, the nuclear norm estimate $\hat{X}_\mathrm{nuc}$ is calculated on a 20-point grid of $\lambda$, logarithmically spaced between 0.01 and 0.1. If multiple solutions of rank $r$ exist, the one with the smallest $\lambda$ value is selected.

The boxplots of estimation errors and computation time are illustrated in Fig.~\ref{fig:3}. It can be seen that the LAR solution gives closer estimates compared to the other two methods, with median prediction error reductions of 28\% and 38\% compared to nuclear norm regularization and LS-TSVD, respectively. 
The computation time of the LAR solution is similar to \textit{LS-TSVD} and significantly shorter than nuclear norm regularization.
\begin{figure}[tb]
    \centering
    \includegraphics[scale=0.93]{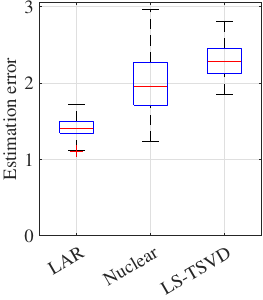}
    \includegraphics[scale=0.93]{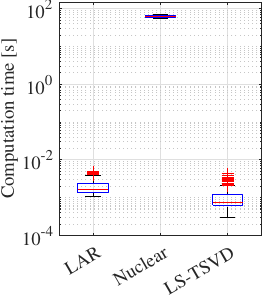}
    \caption{Comparison of estimation errors and computation time for unstructured matrix regression. LAR: closed-form solution \eqref{eq:xlar}, Nuclear: nuclear norm regularization \eqref{eq:nuc}, LS-TSVD: successive least-squares estimate and truncated SVD.}
    \label{fig:3}
\end{figure}

For Hankel matrix regression, Example~2 in Section~\ref{sec:exam} is considered with a sixth-order ($r=6$) system:
\begin{equation}
    G(q)=d_0\sum_{i=1}^{n_p} \left(\frac{d_i}{q-q_i}+\frac{d_i^*}{q-q_i^*}\right),
    \label{eq:model}
\end{equation}
where $n_p=3$, $q_1=-0.6+0.6j$, $q_2=0.9+0.2j$, $q_3=0.2+0.9j$, $d_1=1$, $d_2=1-2j$, $d_3=-1-j$, and $d_0$ is selected such that $G(q)$ has an $\mathcal{H}_2$-norm of 1. The impulse response of $G(q)$ is given by $g_k=d_0\sum_{i=1}^{n_p}\left(d_i q_i^{k-1}+d_i^* \left(q_i^*\right)^{k-1}\right)$, so $q_i$ and the angle of $d_i$ correspond to the optimal choices of $z_i$ and $\psi_i$, respectively.
The following parameters are used in simulations: $m=80$, $n=20$, $\epsilon=0.01$. Two noise levels are considered with a standard deviation of 0.01 and 0.1 for $e_k$, respectively. The nuclear norm estimates are obtained similarly to the previous example, except that the $\lambda$-grid is selected between 0.1 and 1. We also compare two other algorithms for structured low-rank matrix regression in existing works, namely Cadzow's algorithm \cite{Cadzow_1988} (\textit{Cadzow}) and the SLRA package \cite{Markovsky_2014} (\textit{SLRA}). It is assumed that the system is known to be stable, so $z_i$ is constrained by $|z_i|\leq 1$ in Algorithm~\ref{al:2} as discussed in Remark~\ref{rm:1}. Note that this stability constraint cannot be guaranteed for the other algorithms.

Table~\ref{tbl:1} shows the estimated values of $z_i$, $\psi_i$ from Algorithm~\ref{al:2} against their optimal values from the true model \eqref{eq:model}. The close match of the values validates the effectiveness of the LAR algorithm. The boxplot of estimation errors is illustrated in Fig.~\ref{fig:2}. The average computation times are 1.813$\,$s, 1.906$\,$s, 55.529$\,$s, 0.153$\,$s, and 14.594$\,$s, for the \textit{LAR}, \textit{LAR-LS}, \textit{Nuclear}, \textit{Cadzow}, and \textit{SLRA} algorithms, respectively. The results demonstrate that the LAR algorithm with the least-squares debiasing \eqref{eq:debias} obtains closer estimates with significantly shorter computation time compared to nuclear norm regularization. The reductions in median estimation errors are 49\% and 70\% for LAR and LAR-LS, respectively. This proves the advantages of using the LAR algorithm over directly solving the nuclear norm regularization problem. Compared to \textit{Cadzow} and \textit{SLRA}, the proposed LAR algorithm obtains smaller estimation errors under the higher noise level. The computational efficiency of \textit{LAR} is also higher compared to \textit{SLRA}.
\begin{table}[tb]
    \renewcommand{\arraystretch}{1.2}
    \centering
    \caption{Comparison of $z_i$ and $\psi_i$ against their optimal values}
    \vspace{-0.5em}
    
    \begin{tabularx}{\columnwidth}{ccccc}
    \hline\hline
    \multicolumn{2}{c}{\textbf{Poles}} & \textbf{1} & \textbf{2} & \textbf{3}\\\hline
    \multirow{2}{*}{$|z_i|$} & \textbf{Estimated} & $0.865\pm0.012$ & $\phantom{-}0.921\pm0.002$ & $\phantom{-}0.923\pm0.004$\\
     & \textbf{True} & $0.849$ & $\phantom{-}0.922$ & $\phantom{-}0.922$ \\\hline
    \multirow{2}{*}{$\theta_i$} & \textbf{Estimated} & $2.339\pm0.023$ & $\phantom{-}0.214\pm0.003$ & $\phantom{-}1.354\pm0.004$ \\
     & \textbf{True} & $2.356$ & $\phantom{-}0.219$ & $\phantom{-}1.352$ \\\hline
    \multirow{2}{*}{$\psi_i$} & \textbf{Estimated} & $0.153\pm0.115$ & $-1.030\pm0.031$ & $-2.390\pm0.037$ \\
     & \textbf{True} & 0 & $-1.107$ & $-2.356$ \\\hline\hline
    \end{tabularx}
    \label{tbl:1}
\end{table}
\begin{figure}[tb]
    \centering
    \includegraphics[scale=0.93]{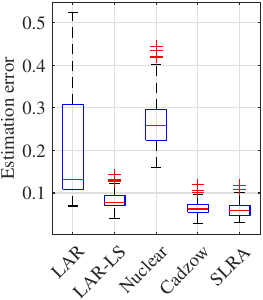}
    \includegraphics[scale=0.93]{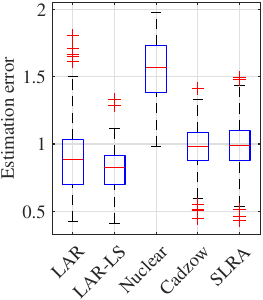}
    \footnotesize (a) Noise level: 0.01\qquad\qquad\qquad\quad (b) Noise level: 0.1
    \caption{Comparison of estimation errors for Hankel matrix regression. LAR: Algorithm~\ref{al:2}, LAR-LS: Algorithm~\ref{al:2} plus least-squares debiasing \eqref{eq:debias}, Nuclear: nuclear norm regularization \eqref{eq:nuc}, Cadzow: Cadzow's algorithm \cite{Cadzow_1988}, SLRA: the SLRA package \cite{Markovsky_2014}.}
    \label{fig:2}
\end{figure}

\section{CONCLUSIONS}

In this work, the low-rank matrix regression problem is reformulated as an infinite-dimensional sparse learning problem and solved by least-angle regression (LAR). When the matrix is unstructured, the LAR algorithm admits a closed-form solution, only requiring two singular value decomposition operations. For Hankel matrices, LAR is implemented by considering a real-valued reformulation of polynomial bases. Significant improvements in both accuracy and efficiency are observed numerically for both cases, compared to nuclear norm regularization. Further works include extending the procedure to other matrix structures and further improving the computational efficiency of infinite-dimensional LAR algorithms.

The results in this work prove that LAR can provide more efficient and effective algorithms than directly solving the nuclear norm regularization problem.

\bibliographystyle{IEEEtran}
\bibliography{refs}

\end{document}